\documentclass[11pt]{article} 

\usepackage[a4paper]{geometry}
\usepackage{t1enc}
\usepackage{url}
\usepackage[latin1]{inputenc}
\usepackage{theorem}
\usepackage[english]{babel}
\usepackage{graphicx}
\usepackage{mathrsfs}
\usepackage{todonotes}
\usepackage{microtype}

\usepackage{bm,amsmath,amssymb,theorem}
\usepackage[all,poly,dvips]{xy}
\usepackage{xcolor}

\newcommand{\adj}{\,\text{\textemdash}\,}
\def\..{\,\mathpunct{\ldotp\ldotp}} 

\newcommand{\noarg}{-} 

\newcommand{\lift}[2]{\widetilde{#1}^{#2}}

\newcommand{\la}{\langle}
\newcommand{\ra}{\rangle}

\renewcommand{\phi}{\varphi}
\renewcommand{\epsilon}{\varepsilon}

\newtheorem{definition}{Definition}
\newtheorem{theorem}{Theorem}

\newcounter{noqed}
\newcommand{\qed}{ \ifmmode\mbox{ }\fi\rule[-.05em]{.3em}{.7em}\setcounter{noqed}{0}}
\newenvironment{proof}[1][{}]{\noindent{\bf Proof#1. }\setcounter{noqed}{1}}{\ifnum\value{noqed}=1\qed\fi\par\medskip}

\begin{document}
\bibliographystyle{alpha}

\title{Spectral Rank Monotonicity on Undirected Networks}%
\author{Paolo Boldi\and Flavio Furia\and Sebastiano Vigna\\
Dipartimento di Informatica, Universit\`a degli Studi di Milano, Italy}

\maketitle

\begin{abstract}
We study the problem of \emph{score} and \emph{rank} monotonicity for \emph{spectral ranking} methods,
such as eigenvector centrality and PageRank, in the case of undirected networks. Score monotonicity means that
adding an edge increases the score at both ends of the edge. Rank monotonicity means that adding
an edge improves the relative position of both ends of the edge with respect to the remaining nodes. 
It is known that common spectral rankings are both score and rank monotone on directed, strongly connected graphs.
We show that, surprisingly, the
situation is very different for undirected graphs, and in particular that PageRank is
neither score nor rank monotone.
\end{abstract}

\section{Introduction}

The study of centrality in networks goes back to the late forties. Since then,
several measures of centrality with different properties have been published (see~\cite{BoVAC} for a survey).
To sort out which measures are more apt for a specific application, one can try
to classify them by means of some axiom that they might satisfy or not. 

In a previous paper~\cite{BLVRMCM}, two of the authors have studied in particular
\emph{score monotonicity}~\cite{BoVAC} and \emph{rank monotonicity} on directed graphs.
The first property says that when an arc
$x\to y$ is added to the graph, the score of $y$ strictly increases.
Rank monotonicity~\cite{CDKLEAA}
states that after adding an arc $x\to y$, all nodes with a score
smaller than or equal to $y$ have still a score smaller than or equal to $y$.
Score and rank monotonicity complement themselves. Score
monotonicity tells us that ``something good happens''. Rank monotonicity
that ``nothing bad happens''. 

Once we move to undirected graphs, however, previous definitions and results are no longer
applicable. Note that adding a single edge to an undirected graph is equivalent to adding \emph{two} opposite arcs
in a directed graph, which may suggest why the situation is so different. In this paper, we
propose definitions that are natural extensions of the directed case, and prove
results about classical types of spectral ranking~\cite{VigSR}---eigenvector
centrality~\cite{LanZRWT,BerTGA}, Seeley's
index~\cite{SeeNRI}, and PageRank~\cite{PBMPCR}. With minor restrictions,
all these measures of centrality have been proven to be score and rank monotone in the directed case~\cite{BLVRMCM}.
However, we will prove that, surprisingly, this is no longer true in the undirected
case: in the case of eigenvector centrality and PageRank, at least one of the extremes of the edge might
lower both its score and its rank.

To prove general results in the case of PageRank, we use the theory of graph fibrations~\cite{BoVGF},
which makes us able to reduce the computation of a spectral ranking of a graph of variable size to a similar
computation on a finite graph.
This approach to proofs, which we believe is of independent interest,
makes it possible to use analytic techniques to control the PageRank values.

We conclude the paper with some anecdotal evidence from a medium-sized real-world
network, showing that violations of rank monotonicity do happen.

\section{Graph-theoretical preliminaries}
\label{sec:defs}

While we will focus on simple undirected graphs, we are going to make use of
some proof techniques that require handling more general types of graphs.

A \emph{(directed multi)graph} $G$ is defined by a set $N_G$ of nodes, a set
$A_G$ of arcs, and by two functions $s_G,t_G:A_G\to N_G$ that specify the
source and the target of each arc (we shall drop the subscripts whenever no
confusion is possible); a \emph{loop} is an arc with the same source and target.
We use $G(i,j)$ for denoting the set of
arcs from node $i$ to node $j$, that is, the set of arcs $a\in A_G$ such that $s(a)=i$ and
$t(a)=j$; the arcs in $G(i,j)$ are said to be \emph{parallel} to one
another.  
Similarly, we denote with $G(\noarg,i)$ the set of arcs coming
into $i$, that is, the set of arcs $a\in A_G$ such that $t(a)=i$, and analogously with
$G(i,\noarg)$ the set of arcs going out of $i$. 
Finally, we write $d_G^+(i)=|G(i,\noarg)|$ for the \emph{outdegree}
of $i$ in $G$ and $d_G^-(i)=|G(\noarg,i)|$ for the \emph{indegree} of $i$ in
$G$.

The main difference between this definition and the standard definition of a
directed graph is that we allow for the presence of multiple arcs between any
pair of nodes.
Since we do not need to distinguish between graphs that only differ because of node names, we will
always assume that $N_G=\{\,0,1,\dots,n_G-1\,\}$ where $n_G$ is the number of nodes of $G$.
Every graph $G$ 
has an associated $n \times n$ \emph{adjacency matrix}, also denoted by $G$, where $G_{ij}=|G(i,j)|$. 

A \emph{(simple) undirected graph} is a loopless\footnote{Note that our negative results are \emph{a fortiori} true if we consider
undirected graphs with loops. Our positive results are still valid in the same case
using the standard convention that loops increase the degree by two.} graph $G$ such that for all $i,j \in N$,
$|G(i,j)|=|G(j,i)|\leq 1$. In other words, there are no parallel arcs and if there is an
arc from $i$ to $j$ there is also an arc in the opposite direction.
In an undirected graph, an \emph{edge} is an unordered set of nodes $\{\,i,j\,\}$ (simply denoted by $i \adj j$) such that $|G(i,j)|=1$; 
the set of all edges will be denoted by $E_G$; obviously, the number of edges is exactly half of the number of arcs. 
For undirected graphs,
we prefer to use the word ``vertex'' instead of ``node'', and use $V$ (instead of $N$) for the set of vertices
and $d(x)$ for the degree of a vertex $x$.

\section{Score and rank monotonicity axioms on undirected graphs}

One of the most important notions that researchers have been trying to capture
in various types of graphs is ``node centrality'':
ideally, every node (often representing an
individual) has some degree of influence or importance within the social domain
under consideration, and one expects such importance to be reflected in the
structure of the social network; centrality is a quantitative measure that
aims at revealing the importance of a node.

Formally, a \emph{centrality} (measure or index) is any function $c$ that, given a graph $G$, assigns a  
real number $c_G(x)$ to every node $x$ 	of $G$; countless notions of centrality have been proposed over time, for
different purposes and with different aims; each of them was originally defined only for a specific category of graphs. Later some of
these notions of centrality have been extended to more general classes; in this paper, we shall only consider centralities 
that can be defined
properly on all undirected graphs (even disconnected ones).

Axioms are useful to isolate properties of different centrality measures and make it possible to compare them. One
of the oldest papers to propose this approach is Sabidussi's paper~\cite{SabCIG}, and many other proposals have appeared in the
last two decades.

In this paper we will be dealing with two properties of centrality measures:

\begin{definition}[Score monotonicity]
Given an undirected graph $G$, 
a centrality $c$ is said to be \emph{score monotone on $G$} iff for every pair of non-adjacent vertices $x$ and $y$ we have that
\[
	c_{G'}(x) > c_G(x) \text{\ and\ } c_{G'}(y) > c_G(y),
\]
where $G'$ is the graph obtained adding the new edge $x \adj y$ to $G$.
It is said to be \emph{weakly score monotone on $G$} iff the same property holds, with $\geq$ instead of $>$.
We say that $c$ is \emph{(weakly) score monotone on undirected graphs} iff it is (weakly) score monotone on all 
undirected graphs $G$.
\end{definition}

\begin{definition}[Rank monotonicity]
Given an undirected graph $G$, 
a centrality $c$ is said to be \emph{rank monotone on $G$} iff for every pair of non-adjacent vertices $x$ and $y$ we have that for all vertices $z\neq x,y$
\[
	c_{G}(x) \geq c_G(z) \Rightarrow c_{G'}(x) \geq c_{G'}(z) \text{\ and \ }
	c_{G}(y) \geq c_G(z) \Rightarrow c_{G'}(y) \geq c_{G'}(z),\\
\]
where $G'$ is the graph obtained adding the new edge $x \adj y$ to $G$.
It is said to be \emph{strictly rank monotone\footnote{Note that the published version of this paper~\cite{BFVSRMUN}
contains a slightly different (and mistaken) definition which does not extend correctly the definition given in~\cite{BLVRMCM}.} on $G$} if instead
\[
	c_{G}(x) \geq c_G(z) \Rightarrow c_{G'}(x) > c_{G'}(z) \text{\ and \ }
	c_{G}(y) \geq c_G(z) \Rightarrow c_{G'}(y) > c_{G'}(z).\\
\]
We say that $c$ is \emph{(strictly) rank monotone on undirected graphs} iff it is (strictly) rank monotone on all 
undirected graphs $G$.
\end{definition}

These four properties\footnote{The asymmetric use of strict/weak in the two definitions is for consistency with the previous literature on this topic.} can be studied on the class of all undirected graphs or only on the connected class, giving rise to eight possible ``degrees
of monotonicity'' that every given centrality may satisfy or not. This paper studies these different degrees of monotonicity for three popular 
spectral centrality measures, also comparing the result obtained with the corresponding properties in the directed case. As we shall see, the
undirected situation is quite different.

\section{Eigenvector centrality}

Eigenvector centrality is probably the oldest attempt at deriving a centrality from matrix information: a first
version was proposed by Landau in 1895 for matrices representing the results of chess tournaments~\cite{LanZRWT}, and
it was stated in full generality in 1958 by Berge~\cite{BerG}; it has been rediscovered many times since then.
One considers the adjacency matrix of the graph and computes its left or right dominant eigenvector, which in our case coincide: the result
is thus defined modulo a scaling factor, and if the graph is strongly
connected, the result is unique (modulo the scaling factor).

\begin{figure}
\centering
\includegraphics{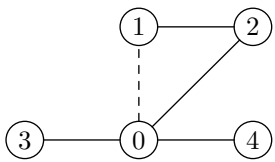}
\caption{\label{fig:ec}A counterexample to score monotonicity for eigenvector centrality. After adding the edge between $0$ and $1$, the score of $0$ decreases:
in norm $\ell_1$, from $0.30656$ to $0.29914$; in norm $\ell_2$, from $0.65328$ to $0.63586$; and when projecting the
constant vector $\mathbf1$ onto the dominant eigenspace, from $1.39213$ to $1.35159$.}
\end{figure}

Discussing score monotonicity requires some form of normalization, due to the presence of the scaling factor. In Figure~\ref{fig:ec} we show a very simple graph violating the property.
In particular, node $0$ score decreases after adding the arc $0\adj1$ both
in norm $\ell_1$ and norm $\ell_2$, and when projecting the constant vector $\mathbf1$ onto the dominant eigenspace, which is an alternative way of circumventing the
scaling factor~\cite{VigSR}. The intuition 
is that initially node $0$ has a high score because of its largest degree (three). However, once we close the triangle
we create a loop that absorbs a large amount of rank, effectively decreasing the score of $0$. We conclude that
\begin{theorem}
Eigenvector centrality does not satisfy weak score monotonicity, even on connected undirected graphs (using norm $\ell_1$, $\ell_2$, or projection onto the dominant eigenspace).
\end{theorem}

\begin{figure}
\centering
\includegraphics{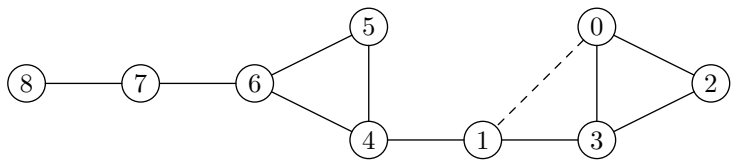}
\caption{\label{fig:ec2}A counterexample to rank monotonicity for eigenvector centrality.
Before adding the edge between $0$ and $1$, the score of $1$ is greater than the score of $3$; after, it is smaller.}
\end{figure}

A similar counterexample, shown in Figure~\ref{fig:ec2}, shows that eigenvector centrality does not satisfy rank monotonicity.
The scores of nodes $3$ and $1$ go from strictly increasing (without the edge $0\adj1$) to strictly decreasing (with the edge $0\adj1$);
thus, $1$ loses rank.
Note that in this case we do not have to choose a normalization, as the order of the two values does not change upon normalization. We conclude that  
\begin{theorem}
Eigenvector centrality does not satisfy weak rank monotonicity, even on connected undirected graphs.
\end{theorem}

\section{Seeley's index}

Seeley's index~\cite{SeeNRI} is simply the steady state of the natural (uniform)
random walk on the graph (for more details, see~\cite{BoVAC}). It is a
well-known fact that if the graph is connected the steady-state probability of
node $x$ is simply $d(x)/2m$---essentially, the index is just the $\ell_1$-normalized degree. We
will thus use this definition for all graphs. As a consequence:

\begin{theorem}
Seeley's index ($\ell_1$-normalized degree) is strictly rank monotone on undirected graphs.
\end{theorem}

The situation is slightly different for score monotonicity:
\begin{theorem}
Seeley's index ($\ell_1$-normalized degree) is score monotone on undirected graphs, except in the case of a disconnected graph formed by a star graph and by one or more additional isolated vertices,
in which case it is just weakly score monotone.
\end{theorem} 
\begin{proof}
When we add an edge between $x$ and $y$ in a graph with $m$ edges, the score of $x$ changes from $d(x)/2m$ to $(d(x)+1)/(2m+2)$. If we require
\[
\frac{d(x)+1}{2m+2}> \frac{d(x)}{2m}
\]
we obtain $d(x)< m$.  Since obviously $d(x)\leq m$, the condition is always true except when $d(x)=m$, which corresponds to
the case of a disconnected graph formed by a star graph and by additional isolated vertices.
Indeed, in that case adding an edge between an isolated vertex and the center of the star will not change the center's Seeley's index.\qed
\end{proof}

\section{Graph fibrations and spectral ranking}
\label{sec:fib}

It is known since seminal works from the '50s in the theory of \emph{graph divisors}~\cite{CDSSG} that fibrations~\cite{BoVGF}, defined below,
have an important relationship with eigenvalues and eigenvectors: if there is a fibration $f:G\to B$, the eigenvalues
of $G$ and $B$ are the same, modulo multiplicity, and eigenvectors of $G$ can be obtained from the eigenvectors of $B$. The
results extend to weighted graphs, too. In this section, we are going to extend such results to \emph{damped spectral rankings}~\cite{VigSR} of the form
\[
\bm v\sum_{i=0}^\infty \beta^iM^i = \bm v( 1 - \beta M)^{-1},
\]
where $M$ is the weighted adjacency matrix of a graph, $\beta$ is a parameter satisfying the condition
$0\leq\beta<1/\rho(M)$, $\rho(M)$ is the spectral radius of $M$, and $\bm v$ is a \emph{preference vector}: 
Katz's index~\cite{KatNSIDSA}, Hubbell's index~\cite{HubIOACI} and PageRank~\cite{PBMPCR} are all examples of damped spectral rankings.

While determining a damped spectral ranking for a \emph{specific} graph essentially requires solving a system of linear equations,
possibly approximating its solution with an iterative method,
doing that for \emph{parametric families} of graphs is tricky and often requires
\emph{ad hoc} approaches. Nonetheless, when the graphs under consideration are sufficiently symmetric, one can try to reduce
the computation using a technique based on fibrations. The idea was introduced in~\cite{BLSGFGIP} for random walks
with restart, and in this section we will extend it to general damped spectral rankings, providing
thus a self-contained (and, in fact, simpler) proof.

Let us start with some additional definitions.
A \emph{path} (of length $n\geq 0$) is a sequence $\pi=\la i_0 a_1 i_1 \cdots
i_{n-1} a_n i_n\ra$, where $i_k\in N_G$, $a_k\in A_G$, $s(a_k)=i_{k-1}$ and
$t(a_k)=i_k$. 
We define $s(\pi)=i_0$ (the \emph{source} of $\pi$), $t(\pi)=i_n$ (the \emph{target} of $\pi$), 
$|\pi|=n$ (the \emph{length} of $\pi$) and let $G^*(i,j) = \{\, \pi
\mid s(\pi)=i, t(\pi)=j\,\}$ (the
set of paths from $i$ to $j$). 


A \emph{(graph) morphism} $f:G\to H$ is given by a pair of functions
$f_N:N_G\to N_H$ and $f_A:A_G\to A_H$ commuting with the source and
target maps, that is, $s_H(f_A(a))=f_N(s_G(a))$ and $t_H(f_A(a))=f_N(t_G(a))$ for all
$a \in A_G$ (again, we shall drop the subscripts whenever no confusion is possible). In other
words, a morphism maps nodes to nodes and arcs to arcs in such a way to
preserve the incidence relation.  
The definition of morphism we give here is the obvious extension to the case of multigraphs of the standard notion the
reader may have met elsewhere.
An \emph{epimorphism} is a morphism $f$ such that both $f_N$ and $f_A$ are surjective.

A \emph{fibration}~\cite{BoVGF} between the graphs $G$ and $B$ is a morphism $f: G\to B$ such
that for each arc $a\in A_B$ and each node $i\in N_G$ satisfying
$f(i)=t(a)$ there is a unique arc $\lift ai\in A_G$ (called the \emph{lifting of
$a$ at $i$}) such that $f(\lift ai)=a$ and $t(\lift ai)=i$.
If $f:G\to B$ is a fibration, $G$
is called the \emph{total graph} and $B$ the \emph{base} of $f$. 
We shall also say that $G$ is \emph{fibered (over $B$)}. The \emph{fiber over a
node $h\in N_B$} is the set of nodes of $G$ that are mapped to $h$, and shall
be denoted by $f^{-1}(h)$. 

A more geometric way of interpreting the definition of
fibration is that given a node $h$ of $B$ and a path $\pi$ terminating at $h$,
for each node $i$ of $G$ in the fiber of $h$ there is a unique path terminating
at $i$ that is mapped to $\pi$ by the fibration; this path is called the
\emph{lifting of $\pi$ at $i$}. 

In Figure~\ref{fig:exfib}, we show two graph morphisms; the morphisms are
implicitly described by the colors on the nodes. The morphism displayed on the
left is not a fibration, as the loop
on the base has no counterimage ending at the lower gray node, and
moreover the other arc has two counterimages with the same target. The
morphism displayed on the
right, on the contrary, is a fibration. Observe that loops are not necessarily
lifted to
loops.

\begin{figure}[htbp]
  \begin{center}
	\includegraphics{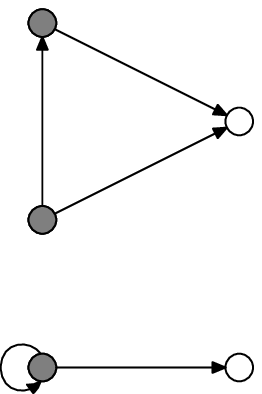}\qquad\qquad\qquad\qquad\includegraphics{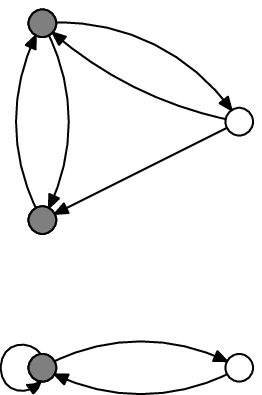}
  \end{center}
  \caption{\label{fig:exfib}On the left, an example of graph morphism that is
not a fibration; on the right, a fibration. Colors on the nodes are used to
implicitly specify the morphisms.}
\end{figure}

We will now show how fibrations can be of help in the computation of a damped spectral ranking.
First of all, we are now going to consider \emph{weighted} graphs, in which each arc is assigned
a real weight, given by a weighting function $w:A_G\to \mathbf R$: the adjacency matrix associated to
a weighted graph $G$  is defined by letting
\[
	G_{ij}=\sum_{a \in G(i,j)} w(a)
\]
and we obtain the unweighted case when $w$ is the constant $a \mapsto 1$ function.
All the morphisms (especially, fibrations) between weighted graphs are assumed to preserve weights. 

Note that every morphism $f: G \to B$ extends to a mapping $f^*$ between paths of $G$ and paths of $B$
in an obvious way. This map $f^*$ preserves not only path lengths, but also weight sequences if $f$ does.

For fibrations we can say more; using the lifting property, one can prove by induction that:
\begin{theorem}
\label{thm:bij}
If $f: G \to B$ is an epimorphic fibration between weighted graphs,
then for every two nodes $j \in N_G$ and $k \in N_B$ the map $f^*$ is a bijection 
between $\cup_{i \in f^{-1}(k)} G^*(i,j)$ and $B^*(k,f(j))$.
\end{theorem}

Now, for every $t\geq 0$, $G^t$ is the matrix whose $ij$ entry contains a summation of contributions, one for each path $\pi \in G^*(i,j)$, 
and the contribution is given by the product of the arc weights found along the way;
hence, by Theorem~\ref{thm:bij}, under convergence assumptions we have that for all $\beta$ and all $i \in N_G$ and $k \in N_B$
\[
	\sum_{i \in f^{-1}(k)}\left(\sum_{t \geq 0} \beta^t G^t\right)_{ij} = \left(\sum_{t \geq 0} \beta^t B^t\right)_{kf(j)},
\]
or equivalently
\[	
\sum_{i \in f^{-1}(k)}\left((1-\beta G)^{-1}\right)_{ij} = \left((1-\beta B)^{-1}\right)_{kf(j)}.
\]
Now, for every vector\footnote{All vectors in this paper are row vectors.} $\bm u$ of size $n_B$, define its \emph{lifting along $f$} as
the vector $\bm u^f$  of size $n_G$ given by
\[
	\left(u^f\right)_i=u_{f(i)}.
\]
For every $j$, we have
\begin{multline*}
	\left(\bm u^f (1-\beta G)^{-1}\right)_j
	=\sum_{i\in N_G} u^f_i \left((1-\beta G)^{-1}\right)_{ij}
	=\\
	=\sum_{k\in N_B}\sum_{i \in f^{-1}(k)} u_{f(i)} \left((1-\beta G)^{-1}\right)_{ij}
	=\sum_{k\in N_B}u_k\left(\sum_{i \in f^{-1}(k)} \left((1-\beta G)^{-1}\right)\right)_{ij}
	=\\
	=\sum_{k\in N_B}u_k\left((1-\beta B)^{-1}\right)_{kf(j)}
	=\left(\bm u (1-\beta B)^{-1}\right)_{f(j)}
\end{multline*}
which can be more compactly written as
\begin{equation}
\label{eqn:resumedsp}
	\bm u^f (1-\beta G)^{-1}=\left(\bm u (1-\beta B)^{-1}\right)^f.
\end{equation}

Equation (\ref{eqn:resumedsp}) essentially states that if we want to compute the damped spectral ranking of the weighted graph $G$,
for a preference vector that is
constant along the fibers of an epimorphic fibration $f: G \to B$, and thus of the form $\bm u^f$, we can compute the damped spectral ranking of the weighted base $B$ 
using $\bm u$ as preference vector, and then lift along $f$ the result. For example, in the
case of Katz's index a simple fibration between graphs is sufficient, as in that case there are no weights to deal with.

%

\paragraph{Implications for PageRank.}
PageRank~\cite{PBMPCR} can be defined as
\[
(1-\alpha)\bm v\sum_{i=0}^\infty \alpha^i\bar G^i = (1-\alpha) \bm v(1-\alpha \bar G)^{-1},
\]
where $\alpha\in[0\..1)$ is the damping factor, 
$\bm v$ is a non-negative preference vector with unit $\ell_1$-norm, 
and $\bar G$ is the row-normalized version\footnote{Here we are assuming that $G$ has no \emph{dangling nodes} (i.e., nodes with outdegree $0$).
If dangling nodes are present, you can still use this definition (null rows are left untouched in $\bar G$), but then to obtain PageRank you 
need to normalize the resulting vector~\cite{BSVPFD,DCGRFPCSLS}. So all our discussion can also be applied
to graphs with dangling nodes, up to $\ell_1$-normalization.} of $G$.

Then $\bar G$ is just the (adjacency matrix of the) weighted version of $G$ defined by letting $w(a)=1/d^+_G(s_G(a))$.
Hence, if you have a weighted graph $B$,
an epimorphic weight-preserving fibration $f: \bar G \to B$, and a vector $\bm u$ of size $n_B$ such that 
$\bm u^f$ has unit $\ell_1$-norm, you can deduce from (\ref{eqn:resumedsp}) that
\begin{equation}
\label{eqn:resumepr}
 (1-\alpha) \bm u^f (1-\alpha \bar G)^{-1}=\left(\bm (1-\alpha) \bm u (1-\alpha B)^{-1}\right)^f.
\end{equation}

On the left-hand side you have the actual PageRank of $G$ for a preference vector that is fiberwise constant;
on the right-hand side you have a spectral ranking of $B$ for the projected preference vector.
Note that $B$ is not row-stochastic, and $\bm u$ has not unit $\ell_1$-norm, 
so technically the right-hand side of equation (\ref{eqn:resumepr}) is
not PageRank anymore, but it is still a damped spectral ranking.

\section{PageRank}

Armed with the results of the previous section, we attack the case of PageRank, which is the most interesting. The first observation
is that
\begin{theorem}
Given an undirected graph $G$ there is a value of $\alpha$ for which PageRank is strictly rank monotone on $G$.
The same is true for score monotonicity, except when $G$ is formed by a star graph and by one or more additional isolated vertices.
\end{theorem}
\begin{proof}
We know that for $\alpha\to 1$, PageRank tends to Seeley's index~\cite{BSVPFDF}. Since Seeley's index is strictly rank monotone, for each non-adjacent pair 
of vertices $x$ and $y$ there is a value $\alpha_{xy}$ such that for $\alpha\geq\alpha_{xy}$ adding the edge $x\adj y$ is strictly
rank monotone. The proof is completed by taking $\alpha$ larger than all $\alpha_{xy}$'s. 
The result for score monotonicity is similar.\qed  
\end{proof}

On the other hand, we will now show that 
\emph{for every possible value of the damping factor $\alpha$} there is a graph on which PageRank violates rank and score monotonicity.

The basic intuition of our proof is that when you connect a high-degree node $x$ with a low-degree node $y$, $y$ will pass to $x$
a much greater fraction of its score than in the opposite direction. This phenomenon is caused by the stochastic normalization of the
adjacency matrix: the arc from $x$ to $y$ will have a low coefficient, due to the high degree of $x$, whereas the arc from $y$ to $x$
will have a high coefficient, due to the low degree of $y$.

We are interested in a parametric example, so that we can tune it for different values of $\alpha$. At the same time, we want to 
make the example analytic, and avoid resorting to numerical computations, as that approach would make it impossible to prove a result
valid for every $\alpha$---we would just, for example, prove it for a set of samples in the unit interval.

\begin{figure}
\centering
\begin{tabular}{cc}
\raisebox{.5cm}{$G_k$\qquad}&\includegraphics{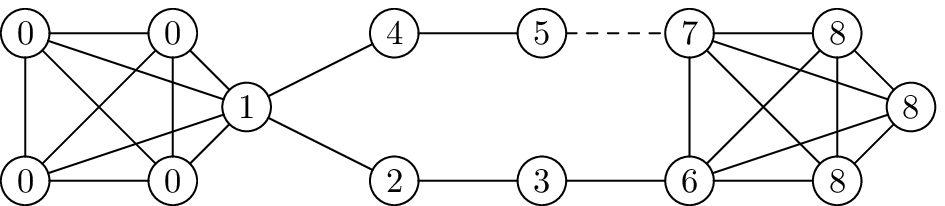}\\
\raisebox{.5cm}{$B_k$\qquad}&\includegraphics{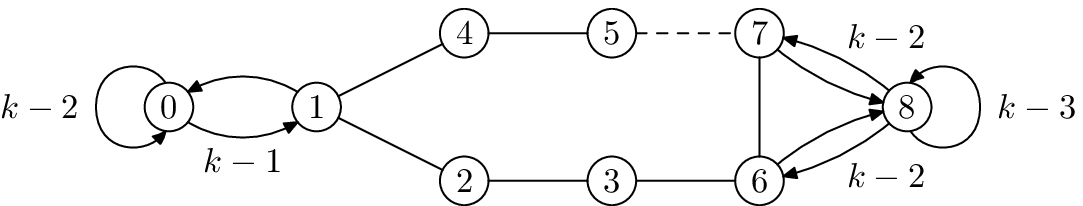}
\end{tabular}
\caption{\label{fig:pr}The parametric counterexample graph for PageRank. The two $k$-cliques are represented here as $5$-cliques for simplicity. Arc
labels represent multiplicity; weights are induced by the uniform distribution on the upper graph.}
\end{figure}

We thus resort to fibrations, using equation (\ref{eqn:resumepr}). In Figure~\ref{fig:pr} we show a parametric
graph $G_k$ comprising two $k$-cliques (in the figure, $k=5$). 
Below, we show the graph $B_k$ onto which $G_k$ can be fibred by mapping
nodes following their labels. The dashed edge is the addition that we will study:
the fibration exists whether the edge exists or not (in both graphs).

While $G_k$ has $2k+4$ vertices, $B_k$ has $9$ vertices, independently of $k$, and
thus its PageRank can be computed analytically as rational functions of $\alpha$ whose coefficients are rational functions
in $k$ (as the number of arcs of each $B_k$ is different). The adjacency matrix of $B_k$ without the dashed arc, considering multiplicities, is\footnote{Note that in the published version of this paper~\cite{BFVSRMUN}
the denominators of the second row are $k-1$, mistakenly, instead of $k+1$.}
\[
\left(\begin{matrix}
    \frac{k-2}{k-1}& \frac{k-1}{k-1}& 0& 0& 0& 0& 0& 0& 0\\
    \frac1{k+1}& 0& \frac1{k+1}& 0& \frac1{k+1}& 0& 0& 0& 0\\
    0&   \frac12& 0& \frac12& 0& 0& 0& 0& 0\\
    0&   0& \frac12& 0& 0& 0& \frac12& 0& 0\\
    0&   \frac12& 0& 0& 0& \frac12& 0& 0& 0\\
    0&   0& 0& 0& 1& 0& 0 &\fcolorbox{gray}{gray}{0}& 0\\
    0&   0& 0& \frac1k& 0& 0& 0& \frac1k& \frac1k\\
    0&   0& 0& 0& 0& \fcolorbox{gray}{gray}{0}&  \frac1{k-1}& 0& \frac1{k-1}\\
    0&   0& 0& 0& 0& 0& \frac{k-2}{k-1}& \frac{k-2}{k-1}& \frac{k-3}{k-1}\\
\end{matrix}\right)
\]
After adding the edge between $5$ and $7$ we must modify the matrix by setting the two grayed 
entries to one and fix normalization accordingly. We will denote with $\operatorname{pre}_\alpha(x)$ the rational function returning the PageRank of
node $x$ with damping factor $\alpha$ before the addition of the dashed arc, and 
with $\operatorname{post}_\alpha(x)$ the rational function returning the PageRank of
node $x$ with damping factor $\alpha$ after the addition of the dashed arc. 

We use the Sage computational engine~\cite{Sage} to perform all computations, as
the resulting rational functions are quite formidable.\footnote{The Sage worksheet can be found at \url{https://vigna.di.unimi.it/pagerank.ipynb}.} 
We start by considering
node $5$: evaluating $\operatorname{post}_\alpha(5)-\operatorname{pre}_\alpha(5)$ in $\alpha=2/3$ we obtain
a negative value for all $k\geq 12$, showing
there is always a value of $\alpha$ for which node $5$ violates weak score monotonicity, as long as $k\geq 12$. 

To strengthen our results, we are now going to show that for \emph{every} $\alpha$ there is a $k$ such that
weak score monotonicity is violated. We use Sturm polynomials~\cite{RaSATP} to compute the number of sign changes of the 
numerator $p(\alpha)$ of  $\operatorname{post}_\alpha(5)-\operatorname{pre}_\alpha(5)$ 
for $\alpha\in[0\..1]$, as the denominator cannot have zeros. Sage
reports that there are two sign changes for $k\geq 12$, which means that $p(\alpha)$ is initially positive; then, somewhere before $2/3$
it becomes negative; and finally it returns positive again somewhere after $2/3$.

Determining the behavior of the points at which $p(\alpha)$ changes sign is impossible due to the high degree of the polynomials
involved. However, we can take two suitable parametric points in the unit interval that sandwich $2/3$, such as 
\[
a = \frac34 - \frac{3k}{4k+ 1000} \leq \frac23 \leq \frac12+\frac{k}{2k+1000} = b,
\]
and use again Sturm polynomials to count the number of sign changes in $[0\..a]$ and $[b\..1]$. In both cases, if $k\geq 15$
there is exactly one sign change in the interval, and since $a\to0$, $b\to 1$ as $k\to \infty$, we conclude that as $k$ grows
the size of the interval of $\alpha$'s in which $p(\alpha)<0$ grows, approaching $[0\..1]$ in the limit. Thus,
\begin{theorem}
For every value of $\alpha\in[0\..1)$, there is an undirected graph for which PageRank violates weak score monotonicity
when $\alpha$ is chosen as damping factor.
\end{theorem}

We now use the same example to prove the lack of rank monotonicity. In this case, we study in a similar way
$\operatorname{pre}_\alpha(5)-\operatorname{pre}_\alpha(2)$, which is positive in $\alpha=2/3$
if $k\geq 14$. Its numerator has two sign changes in the unit interval,
which means that initially $5$ has a smaller PageRank than $2$; then, somewhere before $2/3$
$5$ starts having a larger PageRank than $2$; finally, we return to the initial condition.

 Once again, we sandwich $2/3$ using
\[
\frac34 - \frac{3k}{4k+200}\leq \frac23\leq \frac12 + \frac{k}{2k+200},
\]
and with an argument analogous to the case of score monotonicity we conclude that as $k$ grows the subinterval
of values of $\alpha$ in $[0\..1)$ for which the score of $5$ is greater than the score of $2$ grows up to the whole interval.

Finally, we study $r(\alpha)=\operatorname{post}_\alpha(5)-\operatorname{post}_\alpha(2)$ which is negative in $\alpha=2/3$
if $k\geq 6$, and whose numerator has three sign changes in the unit interval. Once again, we sandwich $2/3$ using
\[
a = \frac1{10} - \frac{k}{10k+2000}\leq \frac23\leq \frac12 + \frac{k}{2k+200} = b.
\]
In this case, there are always two sign changes in $[0\..a]$ and one sign change in $[b\..1]$ for $k\geq 25$, so 
there is a subinterval of values of $\alpha$ in $[0\..1)$ for which the score of $5$ is smaller than the score of $2$ after adding the
edge $5\adj 2$, and this subinterval grows in size up to the whole unit interval as $k$ grows. All in all, we proved that:
\begin{theorem}
\label{th:prrank}
For every value of $\alpha\in[0\..1)$, there is an undirected graph for which PageRank violates rank monotonicity
when $\alpha$ is chosen as damping factor.
\end{theorem}

\section{Experiments on IMDB}

\begin{table}[t]
\renewcommand{\arraystretch}{1.2}
\renewcommand{\tabcolsep}{1ex}
\begin{tabular}{lll}
Score increase & Score decrease & Violations of rank monotonicity\\
\hline
Meryl Streep & Yasuhiro Tsushima & Anne--Mary Brown, Jill Corso,~\ldots\\
Denzel Washington & Corrie Glass & Patrice Fombelle, John Neiderhauser,~\ldots\\
Sharon Stone & Mary Margaret (V) & Dolores Edwards, Colette Hamilton,~\ldots\\
John Newcomb & Robert Kirkham & Brandon Matsui, Evis Trebicka,~\ldots 
\end{tabular}
\vspace*{.5em}
\caption{\label{tab:rank}A few examples of violations of score monotonicity and rank monotonicity in the Hollywood co-starship graph
\texttt{hollywood-2011}. If we add an edge between the actors in the first and second column, the first actor has a score increase, the second actor has a score decrease,
and the actors in the third column, which were less important than the second actor, become more important after the edge addition.}
\end{table}

To show that our results are not only theoretical, we provide a few interesting anecdotal examples from
the PageRank scores ($\alpha=0.85$) of the Hollywood co-starship graph,
whose vertices are actors/actresses in the Internet Movie Database, with an edge connecting them if played in the same movie.
In particular, we used the \texttt{hollywood-2011} dataset from the Laboratory for Web Algorithmics,\footnote{\url{http://law.di.unimi.it/}}
which contains approximately two million vertices and $230$ million edges.

To generate our examples, we picked two actors either at random, or considering
the top $1/10000$ of the actors of the graph in PageRank order and the bottom
quartile, looking for a collaboration that would hurt either actor (or
both).\footnote{Note that for this to happen, the collaboration should be a
two-person production. A production with more people would actually add more
edges.} About $4$\% of our samples yielded a violation of monotonicity, and in
Table~\ref{tab:rank} we report a few funny examples.

It is interesting to observe that in the first three cases it is the less-known actor that loses score (and rank) by the collaboration
with the star, and not the other way round, which is counterintuitive.
In the last case, instead, a collaboration would damage the most important vertex, and
it is an open problem to prove a result analogous to Theorem~\ref{th:prrank} for this case.
We found no case in which both actors would be hurt by the collaboration.

\section{Conclusions}

We have studied score and rank monotonicity for three fundamental kinds of spectral ranking---eigenvector
centrality, Seeley's index, and PageRank. Our results show that except for
Seeley's index on connected graphs, there are always cases in which score and rank monotonicity fail,
contrarily to the directed case, and these failures can be found in real-world graphs.
In particular, for PageRank we can find a counterexample for every value of the damping factor.
Finding such a class of counterexamples for Katz's index~\cite{KatNSIDSA} is an interesting open problem.
Another valuable contribution would be to find another class of counterexamples for PageRank that is amenable to a simpler analytic proof without having to rely on computer algebra. 

Our results suggest that common knowledge about the behavior of PageRank in the directed case cannot
be applied automatically to the undirected case.


\vspace*{-.7em}
\bibliography{biblio}

\newcommand{\etalchar}[1]{$^{#1}$}
\hyphenation{ Vi-gna Sa-ba-di-ni Kath-ryn Ker-n-i-ghan Krom-mes Lar-ra-bee
  Pat-rick Port-able Post-Script Pren-tice Rich-ard Richt-er Ro-bert Sha-mos
  Spring-er The-o-dore Uz-ga-lis }
\begin{thebibliography}{PBMW98}

\bibitem[Ber58]{BerTGA}
Claude Berge.
\newblock {\em Th\'eorie des graphes et ses applications}.
\newblock Dunod, Paris, France, 1958.

\bibitem[Ber85]{BerG}
Claude Berge.
\newblock {\em {Graphs}}.
\newblock North--Holland, Amsterdam, 1985.

\bibitem[BFV22]{BFVSRMUN}
Paolo Boldi, Flavio Furia, and Sebastiano Vigna.
\newblock Spectral rank monotonicity on undirected networks.
\newblock In Rosa~Maria Benito, Chantal Cherifi, Hocine Cherifi, Esteban Moro,
  Luis~M. Rocha, and Marta Sales-Pardo, editors, {\em Complex Networks \& Their
  Applications X}, volume 1014 of {\em Studies in Computational Intelligence},
  pages 234--246. Springer, 2022.

\bibitem[BLSV06]{BLSGFGIP}
Paolo Boldi, Violetta Lonati, Massimo Santini, and Sebastiano Vigna.
\newblock Graph fibrations, graph isomorphism, and {P}age{R}ank.
\newblock {\em RAIRO Inform. Th{\'e}or.}, 40:227--253, 2006.

\bibitem[BLV17]{BLVRMCM}
Paolo Boldi, Alessandro Luongo, and Sebastiano Vigna.
\newblock Rank monotonicity in centrality measures.
\newblock {\em Network Science}, 5(4):529--550, 2017.

\bibitem[BSV05]{BSVPFDF}
Paolo Boldi, Massimo Santini, and Sebastiano Vigna.
\newblock Page{R}ank as a function of the damping factor.
\newblock In {\em Proc. of the Fourteenth International World Wide Web
  Conference (WWW 2005)}, pages 557--566, Chiba, Japan, 2005. ACM Press.

\bibitem[BSV09]{BSVPFD}
Paolo Boldi, Massimo Santini, and Sebastiano Vigna.
\newblock Page{R}ank: {F}unctional dependencies.
\newblock {\em ACM Trans. Inf. Sys.}, 27(4):1--23, 2009.

\bibitem[BV02]{BoVGF}
Paolo Boldi and Sebastiano Vigna.
\newblock Fibrations of graphs.
\newblock {\em Discrete Math.}, 243:21--66, 2002.

\bibitem[BV14]{BoVAC}
Paolo Boldi and Sebastiano Vigna.
\newblock Axioms for centrality.
\newblock {\em Internet Math.}, 10(3-4):222--262, 2014.

\bibitem[CDK{\etalchar{+}}04]{CDKLEAA}
Steve Chien, Cynthia Dwork, Ravi Kumar, Daniel~R. Simon, and D.~Sivakumar.
\newblock Link evolution: {A}nalysis and algorithms.
\newblock {\em Internet Math.}, 1(3):277--304, 2004.

\bibitem[CDS78]{CDSSG}
Drago\v{s}~M. Cvetkovi\'c, Michael Doob, and Horst Sachs.
\newblock {\em Spectra of Graphs}.
\newblock Academic Press, 1978.

\bibitem[DCGR06]{DCGRFPCSLS}
Gianna Del~Corso, Antonio Gull\`i, and Francesco Romani.
\newblock Fast {P}age{R}ank computation via a sparse linear system.
\newblock {\em Internet Math.}, 2(3):251--273, 2006.

\bibitem[Hub65]{HubIOACI}
Charles~H. Hubbell.
\newblock An input-output approach to clique identification.
\newblock {\em Sociometry}, 28(4):377--399, 1965.

\bibitem[Kat53]{KatNSIDSA}
Leo Katz.
\newblock A new status index derived from sociometric analysis.
\newblock {\em Psychometrika}, 18(1):39--43, 1953.

\bibitem[Lan95]{LanZRWT}
Edmund Landau.
\newblock {Z}ur relativen {W}ertbemessung der {T}urnierresultate.
\newblock {\em Deutsches {W}ochenschach}, 11:366--369, 1895.

\bibitem[PBMW98]{PBMPCR}
Lawrence Page, Sergey Brin, Rajeev Motwani, and Terry Winograd.
\newblock The {P}age{R}ank citation ranking: Bringing order to the web.
\newblock Technical Report SIDL-WP-1999-0120, Stanford Digital Library
  Technologies Project, Stanford University, 1998.

\bibitem[RS02]{RaSATP}
Qazi~Ibadur Rahman and Gerhard Schmeisser.
\newblock {\em Analytic theory of polynomials}.
\newblock Number~26 in London Mathematical Society New Series. Clarendon Press,
  2002.

\bibitem[Sab66]{SabCIG}
G.~Sabidussi.
\newblock The centrality index of a graph.
\newblock {\em Psychometrika}, 31(4):581--603, 1966.

\bibitem[See49]{SeeNRI}
John~R. Seeley.
\newblock The net of reciprocal influence: {A} problem in treating sociometric
  data.
\newblock {\em Canadian Journal of Psychology}, 3(4):234--240, 1949.

\bibitem[{The}18]{Sage}
{The Sage Developers}.
\newblock {\em {S}ageMath, the {S}age {M}athematics {S}oftware {S}ystem
  ({V}ersion 8.0)}, 2018.

\bibitem[Vig16]{VigSR}
Sebastiano Vigna.
\newblock Spectral ranking.
\newblock {\em Network Science}, 4(4):433--445, 2016.

\end{thebibliography}

\end{document}